\documentclass[%
 preprint,
 amsmath,amssymb,
 aps,
]{revtex4-1}

\usepackage{graphicx}
\usepackage{dcolumn}
\usepackage{bm}
\usepackage{amsthm} %
\usepackage{amscd} %
\usepackage{mathrsfs} 
\newtheorem{thm}{Theorem}[section]

\newtheorem{lem}[thm]{Lemma}

\newtheorem{defn}[thm]{Definition}

\begin{document}


\title{Supersymmetry breakdown for  an  extended version\\ of the 
 Nicolai supersymmetric fermion lattice  model}

\author{Hajime Moriya}
\affiliation{%
Faculty of Mechanical Engineering, Institute of Science and Engineering, Kanazawa University,
Kakuma, Kanazawa 920-1192, Japan
}%


\date{\today}

\newcommand{\cstar}{{{C}}^{\ast}}%
\newcommand{\diam}{{\rm{diam}}}
\newcommand{\R}{{\mathbb{R}}}%
\newcommand{\Z}{{\mathbb{Z}}}%
\newcommand{\CC}{{\mathbb{C}}}%
\newcommand{\NN}{{\mathbb{N}}}%
\newcommand{\nonum}{\nonumber}%
\newcommand{\Lam}{\Lambda}%
\newcommand{\vareps}{\varepsilon}%
\newcommand{\I}{{\mathrm{I}}}%
\newcommand{\J}{{\mathrm{J}}}%
\newcommand{\K}{{\mathrm{K}}}%
\newcommand{\Al}{\mathcal{A}}%
\newcommand{\Ale}{\Al_{+} }%
\newcommand{\Alo}{\Al_{-}}%
\newcommand{\Aleo}{\Al_{\pm} }%
\newcommand{\AlI}{{\Al}({\I})}%
\newcommand{\AlJ}{{\Al}({\J})}%
\newcommand{\core}{\Al_{\circ}}%
\newcommand{\coree}{{\core}_+}%
\newcommand{\coreo}{{\core}_-}%
\newcommand{\coreeo}{{\core}_\pm}%
\newcommand{\coreoe}{{\core}_\mp}%
\newcommand{\AlIe}{\AlI_{+}}%
\newcommand{\AlIo}{\AlI_{-}}%
\newcommand{\AlIeo}{\AlI_{\pm}}%
\newcommand{\AlJe}{\AlJ_{+}}%
\newcommand{\AlJo}{\AlJ_{-}}%
\newcommand{\AlJeo}{\AlJ_{\pm}}%
\newcommand{\cicr}{c_i^{\ast}}%
\newcommand{\ci}{c_i}%
\newcommand{\cjcr}{c_j^{\ast}}%
\newcommand{\cj}{c_j}%
\newcommand{\cjpocr}{c_{j+1}^{\ast}}%
\newcommand{\cjpo}{c_{j+1}}%
\newcommand{\cjmocr}{c_{j-1}^{\ast}}%
\newcommand{\cjmo}{c_{j-1}}%
\newcommand{\cimcr}{c_{i-1}^{\,\ast}}%
\newcommand{\cim}{c_{i-1}}%
\newcommand{\cipcr}{c_{i+1}^{\,\ast}}%
\newcommand{\cip}{c_{i+1}}%
\newcommand{\cnijonecr}{c_{2j+1}^{\ast}}%
\newcommand{\cnijcr}{c_{2j}^{\ast}}%
\newcommand{\ctwoi}{c_{2i}}
\newcommand{\ctwoicr}{c_{2i}^{\, \ast}}
\newcommand{\ctwoipcr}{c_{2i+1}^{\, \ast}}
\newcommand{\ctwoip}{c_{2i+1}}
\newcommand{\ctwoimcr}{c_{2i-1}^{\, \ast}}
\newcommand{\ctwoim}{c_{2i-1}}
\newcommand{\ctwoj}{c_{2j}}
\newcommand{\ctwojcr}{c_{2j}^{\, \ast}}
\newcommand{\ctwojpcr}{c_{2j+1}^{\, \ast}}
\newcommand{\ctwojp}{c_{2j+1}}
\newcommand{\ctwojmcr}{c_{2j-1}^{\, \ast}}
\newcommand{\ctwojm}{c_{2j-1}}
\newcommand{\ctwokcr}{c_{2k}^{\, \ast}}%
\newcommand{\ctwok}{c_{2k}}%
\newcommand{\ctwokpcr}{c_{2k+1}^{\, \ast}}%
\newcommand{\ctwokp}{c_{2k+1}}%
\newcommand{\ctwokmcr}{c_{2k-1}^{\, \ast}}%
\newcommand{\ctwokm}{c_{2k-1}}%
\newcommand{\ctwolcr}{c_{2l}^{\, \ast}}%
\newcommand{\ctwol}{c_{2l}}%
\newcommand{\ctwolpcr}{c_{2l+1}^{\, \ast}}%
\newcommand{\ctwolp}{c_{2l+1}}%
\newcommand{\ctwolmcr}{c_{2l-1}^{\, \ast}}%
\newcommand{\ctwolm}{c_{2l-1}}%
\newcommand{\zeromap}{{\mathbf{0}}}
\newcommand{\Fp}{F_{+}}%
\newcommand{\Fm}{F_{-}}%
\newcommand{\Ap}{A_{+}}%
\newcommand{\Am}{A_{-}}%
\newcommand{\Fpm}{F_{\pm}}%
\newcommand{\Gp}{G_{+}}%
\newcommand{\Gm}{G_{-}}%
\newcommand{\Gpm}{G_{\pm}}%
\newcommand{\Gpmd}{{\Gpm}^{\ast}}%
\newcommand{\Hil}{\mathscr{H}}%
\newcommand{\Ome}{\Omega}%
\newcommand{\Hilome}{{\Hil}_{\omega}}%
\newcommand{\Hilvp}{{\Hil}_{\varphi}}%
\newcommand{\piome}{\pi_{\omega}}%
\newcommand{\Omeome}{\Ome_{\omega}}%
\newcommand{\pivp}{\pi_{\varphi}}%
\newcommand{\Omevp}{\Ome_{\varphi}}%
\newcommand{\Bl}{{\mathfrak B}}
\newcommand{\BO}{\textrm{O}}%
\newcommand{\lo}{\textrm{o}}%
\begin{abstract}
 Sannomiya-Katsura-Nakayama have  recently  studied  
 an   extension  of  the Nicolai  supersymmetric fermion lattice model 
 which is named  ``the extended Nicolai model''. 
The extended Nicolai model is  parameterized by an adjustable constant  $g\in\R$ in  
 its defining supercharge, and  satisfies   ${{\cal{N}}}=2$  supersymmetry.
We  show  that  for  any non-zero $g$ the  extended Nicolai  model 
   breaks   supersymmetry  dynamically,   and  the  energy density
 of  any homogeneous  ground state for the model  is strictly positive.
\end{abstract}

\pacs{Valid PACS appear here}
\maketitle

\section{Purpose}
\label{sec:INTRO}
In  \cite{SKN} Sannomiya-Katsura-Nakayama  investigated  supersymmetry breakdown  
  for an  extended version of the Nicolai 
supersymmetric fermion lattice model  \cite{NIC}. 
 As this   model   satisfies the  algebraic relation 
 of  ${{\cal{N}}}=2$  supersymmetry, and 
 it is reduced  the original Nicolai model   when 
 its  adjustable parameter $g\in\R$ is equal to $0$, it  is   called 
 {\it{the extended Nicolai model}}.  

 Supersymmetry breakdown  
 for the extended Nicolai model has been shown  for   
any  non-zero   $g$   on finite systems \cite{SKN}.
In the infinite-volume limit,  however, 
Sannomiya-Katsura-Nakayama verified   supersymmetry breakdown of the extended Nicolai model 
  only when  $g> g_0:=4/\pi$.
This  restriction upon the parameter $g$ seems  to be merely due to 
  technical nature  and its  physics meaning  is unclear. 
The purpose of this  note is to remove this    restriction 
 upon  $g$ in the case of  the infinite-volume limit. 
 We   show  that  for any $g\ne 0$ 
the  extended Nicolai  model defined on  $\Z$  breaks supersymmetry 
 dynamically.
Furthermore, we prove that  for any $g\ne 0$  the energy density of  
 any (homogeneous) ground state 
 for  the extended Nicolai model is  strictly positive.

 In \cite{SKN} it is noted    that 
 even if  supersymmetry is broken for any finite subsystem,  
 there may be   restoration of   supersymmetry  
in the infinite-volume limit as  exemplified in \cite{WITT82}.
We  clarify that  such restoration  
  can not  happen for the extended Nicolai model
 by  formulating  the model as  
  supersymmetric $\cstar$-dynamics in our framework  \cite{AHP}. 
   Our  proof based on  $\cstar$-algebraic methods  is rather 
   model-independent. It   makes  essential use of      
  a  crucial finding by Sannomiya-Katsura-Nakayama  
 (Eq.(15) of \cite{SKN})
 that will be   reformulated  in terms of superderivations of an
 infinite-volume  $\cstar$-system.

\section{The extended   Nicolai supersymmetric fermion lattice  model}
\label{sec:MODEL}
We consider  spinless fermions  over an infinitely extended lattice $\Z$.
Let $\ci$ and $\cicr$ denote the annihilation 
operator and the creation operator of a spinless fermion at $i\in \Z$, respectively. Those obey  the canonical anticommutation relations (CARs):
 For  $i,j\in\Z$   
\begin{align}
\label{eq:CAR}
\{ \cicr, \cj \}&=\delta_{i,j}1, \nonumber \\
\{ \cicr, \cjcr \}&=\{ \ci, \cj \}=0.
\end{align}
 The fermion number operator on  each site  $i\in \Z$ 
is given by $n_i:=\cicr \ci$.

For any $g\in\R$ we take the following 
 infinite sum of local  fermion operators:
\begin{align}
\label{eq:Qg}
Q(g):=
\sum_{k\in \Z} (gc_{2k-1}+ c_{2k-1} c^{\ast}_{2k} c_{2+1}).
\end{align}
It is interpeted as  perturbation of  the supercharge of the original Nicolai 
 model $Q(0)$ by another  supercharge $\sum_{k\in \Z} c_{2k-1}$ multiplied by $g$.  
Note that  the perturbed term $\sum_{k\in \Z} c_{2k-1}$ itself 
 generates  a trivial model. 
By  some  formal  computation using  \eqref{eq:CAR} we see  that  
 $Q(g)$ is  nilpotent:
\begin{align}
\label{eq:Qg-nil}
0=Q(g)^2={Q(g)^{\ast}}^{2}.
\end{align}
Let a supersymmetric Hamiltonian  be  given  as 
\begin{equation}
\label{eq:Hgper}
H(g):=\{Q(g),\; Q(g)^{\ast} \}.
\end{equation}
For any $g\in \R$,
 the model has ${{\cal{N}}}=2$  supersymmetry by definition.
As noted  above, if $g=0$, 
 it corresponds to   the supersymmetric fermion lattice  model  defined by  Nicolai in  \cite{NIC}.

We note that  in the infinite-volume  system 
 either $Q(g)$ or  $Q(g)^{\ast}$,  or both 
  can not exist as a well-defined linear  
 operator if  the supersymmetry associated with  them    
   breaks dynamically: In fact we will show that this is the case unless  $g=0$.
Nevertheless,  its supersymmetric dynamics always makes sense 
   in the infinitely extended system  as we will see  later. 

We  shall  consider  the model  under   periodic boundary conditions as 
  in \cite{SKN}. 
Let $M, N\in 2\NN$. Define 
\begin{align}
\label{eq:Qgper}
\widetilde{Q}(g)_{[-M+1,N]}:=
\sum_{k=-M/2+1}^{N/2} \left(gc_{2k-1}+ c_{2k-1} c^{\ast}_{2k} c_{2k+1}\right),
\end{align}
where $N+1$ is  identified with   $-M+1$.
We see that 
\begin{align}
\label{eq:Qgper-nil}
0=\widetilde{Q}(g)_{[-M+1,N]}^2={\widetilde{Q}(g)_{[-M+1,N]}^{\ast\;2}}.
\end{align}
Then we  define the corresponding local supersymmetric Hamiltonian on
 the same region $[-M+1,N]$  as 
\begin{equation}
\label{eq:Hgper}
\widetilde{H}(g)_{[-M+1,N]}:=\Bigl\{\widetilde{Q}(g)_{[-M+1,N]},\; {\widetilde{Q}(g)}_{[-M+1,N]}^{\ast} \Bigr\}.
\end{equation}

Also we may consider free boundary conditions upon  supercharges.  Let 
\begin{align}
\label{eq:Qgfree}
\widehat{Q}(g)_{[-M+1,N+1]}&:=
\sum_{k=-M/2+1}^{N/2} \left(gc_{2k-1}+ c_{2k-1} c^{\ast}_{2k} c_{2k+1}\right)+gc_{N+1}.
\end{align}
We see that 
\begin{align}
\label{eq:Qgper-nil}
0=\widehat{Q}(g)_{[-M+1,N+1]}^2={\widehat{Q}(g)_{[-M+1,N+1]}^{\ast\;2}}.
\end{align}
We give  a local supersymmetric Hamiltonian 
upon the same region  $[-M+1,N+1]$  by  the following  supersymmetric form:
\begin{equation}
\label{eq:Hgfree}
\widehat{H}(g)_{[-M+1,N+1]}:=\Bigl\{\widehat{Q}(g)_{[-M+1,N+1]},\; {\widehat{Q}(g)}_{[-M+1,N+1]}^{\ast} \Bigr\}.
\end{equation}
Note that we have imposed the free boundary condition upon local supercharges 
 rather than   local Hamiltonians. Thus  
$\widehat{H}(g)_{[-M+1,N+1]}$
 differs from the usual  free-boundary Hamiltonian upon $[-M+1,N+1]$ that has 
 more terms near the edges  
 although both of them   are  localized in  $[-M+1,N+1]$ 
 and give rise to 
 the same time evolution   as  we let  $[-M+1,N+1]$   to $\Z$.
(See the next section.)

 \section{Mathematically rigorous formulation}
\label{sec:MODEL}
 In  \cite{AHP}  
   a general framework  of supersymmetric fermion lattice models is given. 
 By using this framework we shall  reformulate  the extended Nicolai  model  
 introduced  in the preceding section
 as supersymmetric C*-dynamics {\footnote{Usually 
 one first considers  finite systems  and
 then takes  their   infinite volume limit.   
 We have a  different viewpoint here.  
 First we are given an  infinitely extended system upon $\Z$
 that represents the total system. The total system  includes  subsystems as its subalgebras.  See \cite{BR} for C*-algebraic   treatment of quantum statistical mechanics.}}.

For each finite subset $\I\Subset\Z$, 
 let $\AlI$ denote  the  finite-dimensional algebra
generated by $\{\ci, \, \cicr\, ;\;i\in \I\}$.
 For    $\I \subset \J\Subset \Z$, $\AlI$ is 
 imbedded into $\AlJ$. We define 
\begin{equation} 
\label{eq:CARloc}
\core:=\bigcup_{\I \Subset \Z }\AlI,
\end{equation}
where all finite subsets $\I$ of $\Z$ are taken.
The  norm completion of the  $\ast$-algebra $\core$ (with the  operator norm)
 yields  a   $\cstar$-algebra $\Al$ which is   known as   the CAR algebra.
 The dense subalgebra $\core$  is usually   called the local algebra.

Let $\sigma$ denote the shift-translation  automorphism group on $\Al$:
 For  each  $k\in\Z$
\begin{align}
\label{eq:sigk}
\sigma_{k} (c_{i})=c_{i+k},\quad 
\sigma_{k} (c_{i}^{\ast})=c_{i+k}^{\ast},
\quad \forall i  \in \Z.
\end{align}

Let $\gamma$ denote the grading automorphism on the $\cstar$-algebra  $\Al$ determined  by    
 \begin{equation}
\label{eq:CARgamma}
\gamma(\ci)=-\ci, \quad \gamma(\cicr)=-\cicr,\quad \forall i\in \Z.
\end{equation}
The total system $\Al$ is decomposed into the even part and the odd part:  
\begin{align}
\label{eq:grad}
\Al&=\Ale\oplus \Alo,\quad 
\Ale= \{A\in \Al| \; \gamma(A)=A\},\quad
\Alo= \{A\in \Al|\;  \gamma(A)=-A\}.
\end{align}
For each  $\I\Subset\Z$ 
\begin{equation}
\label{eq:CARIeo}
\AlI=\AlIe\oplus\AlIo,\ \ \text{where}\ \ 
 \AlIe := \AlI\cap \Ale,\quad  \AlIo := \AlI\cap \Alo.
 \end{equation}
The  graded commutator   is  defined as  
\begin{align}
\label{eq:gcom}
[\Fp, \;  G]_{\gamma} &= [\Fp, \;  G]
  {\mbox {\ \ for \ }}
\Fp \in \Ale, \ G \in \Al, \nonum\\ 
[\Fm, \; \Gm]_{\gamma} &= \{\Fm,  \; \Gm\} {\mbox {\ \ for \ }}
\Fm \in \Alo, \ \Gm \in \Alo. 
\end{align}
From   the  canonical anticommutation relations 
 \eqref{eq:CAR}   the graded-locality follows:  
\begin{equation}
\label{eq:gloc}
[A,\; B]_{\gamma}=0{\text {\ \ for all  \ }} A \in \AlI 
{\text {\ and \ }}B\in \AlJ {\text {\ \ if   \ }} \I\cap\J=\emptyset,\ \I,\J\Subset\Z.
\end{equation}

The notation   $Q(g)$  of  Eq.\eqref{eq:Qg}  is a  merely formal expression 
 of  the supercharge.
Nevertheless  it gives  a well-defined  infinitesimal fermionic transformation.
Define  a superderivation (a linear map that satisfies 
the graded Leibniz rule) from $\core \to\core$ by 
\begin{equation}
\label{eq:delg}
\delta_g(A):=[Q(g),\; A]_{\gamma}
{\text {\ \ for every  \ }} A \in \core.
\end{equation}
 Similarly the  conjugate superderivation  is given by
\begin{equation}
\label{eq:delgast}
\delta_g^{\ast}(A):=[Q(g)^{\ast},\; A]_{\gamma}
{\text {\ \ for every  \ }} A \in \core. 
\end{equation}
Let us explain   the formula  \eqref{eq:delg} in some depth.
For each  {\it{fixed}} local element $A\in \core$, 
 only finite terms  in the summation formula of 
  $Q(g)$  are involved in $[Q(g),\; A]_{\gamma}$, because 
 there is a least  $\I\Subset\Z$ 
 (with respect to  the inclusion)   such that $A\in \AlI$, and 
 by  the graded-locality \eqref{eq:gloc}
 only  local fermion terms of \eqref{eq:Qg}
that have non-trivial intersection with $\I$ 
 may contribute to  $[Q(g),\; A]_{\gamma}$. 
The other infinite number of  terms give zero.
Therefore   there exist $M_{0}, N_{0}\in2\NN$
 such that $[-M_{0}+1,N_{0}]\supset \I$ and 
 the  identity   
\begin{equation}
\label{eq:delg-local}
\delta_g(A)=\left[\widetilde{Q}(g)_{[-M+1,N]},\; A\right]_{\gamma}
\end{equation}
holds  for all  $M(\in 2\NN)\ge M_0$ and $N(\in 2\NN) \ge N_0$, 
 where  the  periodic-boundary condition as in   \eqref{eq:Qgper} is  used. 
For example,  if $\I$ is a finite interval of the type 
 $[-S+1,-S+2, \cdots,T-1, T]$ with $(S,T\in2\NN)$, 
then it is enough to take $M_{0}=S+2,N_{0}=T+2$.
An  analogous identity by   free-boundary supercharges 
 \eqref{eq:Qgfree} is possible.
For each $A\in \core$ we have
 the asymptotic formula 
\begin{equation}
\label{eq:asympN}
\delta_g(A)=\lim_{N\to\infty}\left[\widetilde{Q}(g)_{[-N+1,N]},\; A\right]_{\gamma},
\end{equation}
and similarly 
\begin{equation}
\label{eq:asympNrfee}
\delta_g(A)=\lim_{N\to\infty}\left[\widehat{Q}(g)_{[-N+1,N+1]},\; A\right]_{\gamma}.
\end{equation}
 The nilpotent condition \eqref{eq:Qg-nil} is  expressed by
the superderivation $\delta_g$ as 
\begin{equation}
\label{eq:delnilpotent}
\delta_g\circ \delta_g=
\delta_g^{\ast}\circ \delta_g^{\ast}=\zeromap.
\end{equation}

Define
 the derivation generated by  the Hamiltonian $H(g)$. 
\begin{equation}
\label{eq:Derh-intro}
d_g(A):=[H(g),\; A]
{\text {\ \ for every  \ }} A \in \core. 
\end{equation}
 This is   the  infinitesimal time-generator of  the model.  
We  can  immediately  verify the following  supersymmetric relation:
\begin{equation}
\label{eq:del-SUSY-relations}
d_g(A)=
\delta_g^{\ast}\circ \delta_g(A)+
\delta_g \circ \delta_g^{\ast}(A)
{\text {\ \ for every  \ }} A \in \core. 
\end{equation}
It has been known   that  short-range interactions of
  fermion lattice systems    give   Hamiltonian dynamics in the 
infinite-volume limit  \cite{BR}.  
Somewhat heuristically, 
we have   for any $A\in\Al$ and $t\in\R$
\begin{align*}
\alpha_{g}(t)(A)
:=\lim_{N\to\infty}\exp\left(it \widetilde{H}(g)_{[-N+1,N]}\right)A
\exp\left(-it \widetilde{H}(g)_{[-N+1,N]}\right)\in\Al,
\end{align*} 
where  special 
 local Hamiltonians given in  \eqref{eq:Hgper}  are used for concreteness.
However,  we may take  any  boundary condition upon local Hamiltonians on 
 finite  subsystems. 
 
We can  construct    supersymmetric dynamics 
 in the infinitely extended  system corresponding to the extended Nicolai model 
  as in the following theorem. Note that 
 it holds irrespective of broken-unbroken supersymmetry.
\begin{thm}
\label{thm:DYN}
For each $g\in\R$
 the superderivation $\delta_g$ generates  a supersymmetric dynamics in $\Al$.
 Precisely,  there exists  a strongly continuous one parameter
group of $\ast$-automorphisms $\alpha_{g}(t)$ ($t\in \R$)
on  $\Al$ whose pre-generator  is given by the derivation  $d_g\equiv 
\delta_g^{\ast}\circ  \delta_g+\delta_g \circ \delta_g^{\ast}$  
 on the local algebra  $\core$. 
\end{thm}
\begin{proof}
From  our  work  \cite{AHP}
 the statement follows immediately.
\end{proof}

We need to    fix  the crucial   terminology `` supersymmetry(SUSY) breakdown'' for the present paper.
In physics literature,   SUSY breakdown 
 is usually identified with   strict positivity of  SUSY Hamiltonian, see e.g.
  \cite{WEIN}.
However, one should be caution 
when dealing with  models on  non-compact  space.   
(In fact,  we  see  a relevant remark  by Witten \cite{WITT82}.)
As shown in   Theorem \ref{thm:DYN}     
   superderivations  are building blocks for  supersymmetric  dynamics.
 So let us introduce  the  following  definition which is 
 based on  invariance under  superderivations.
 We consider that  it is  a  straightforward expression  of the  physics concept  
 of symmetry and symmetry breakdown. 
\begin{defn}
\label{defn:broken}
Suppose that  a  superderivation  generates a supersymmetric dynamics 
 as  in   Theorem \ref{thm:DYN}.
If a state of $\Al$ is invariant under  the superderivation defined on the local system  $\core$,
 then it is called a supersymmetric state.
If there exists 
 no supersymmetric state of $\Al$,  then SUSY is spontaneously broken.
\end{defn}

 Sannomiya-Katsura-Nakayama  employed   a  different  definition  \cite{SKN}:
{\it{SUSY is  spontaneously broken if the energy density of ground states is strictly positive.}}

This alternative  definition based on the energy density seems not  satisfactory in some respects. 
First, it is only limited to homogeneous ground states. 
There may exist  non-periodic  ground states 
 that do not have a well-defined energy density 
as  we have observed such states  for  the original Nicolai model \cite{ergodic} \cite{KMN}.
 It  is  not  obvious  how the  status of SUSY 
 for homogeneous states implies  that for  non-homogeneous states in the infinite-volume limit.
 Second,  its  full justification has not yet been  done 
  even for the particular model  (i.e.  
 the extended Nicolai model)  in  \cite{SKN}.
Actually, we shall discuss the  second point in the next section.

\section{Supersymmetry breakdown}
\label{sec:}
\subsection{Supersymmetry breakdown in the infinite-volume system}
\label{subsec:}
The  first   theorem   is a direct  consequence of 
 a crucial property of the extended Nicolai model \cite{SKN} which is 
 stated below in \eqref{eq:Ok} and \eqref{eq:Okg}.
 We only need to show that  it  is still  valid in the infinite-volume  limit.
\begin{thm}
\label{thm:gbroken}
For any  $g\ne 0$,  
the extended Nicolai supersymmetric fermion lattice model 
 breaks SUSY  spontaneously.
\end{thm}

\begin{proof}
As  given in  Eq.(15) of \cite{SKN}, 
for each $k\in\Z$ let
\begin{align}
\label{eq:Ok}
O_k:=
c_{2k-1}^{\, \ast}
\left(1-\frac{1}{g}\left(c_{2k}^{\, \ast}c_{2k+1}+
c_{2k-3}c_{2k-2}^{\, \ast}\right)
+\frac{2}{g^2}c_{2k-3}
c_{2k-2}^{\, \ast}
c_{2k}^{\, \ast}c_{2k+1}\right).
\end{align}
We shall show that for all $k\in\Z$
\begin{align}
\label{eq:Okg}
\delta_g(O_k)=g.
\end{align}
As the model is $\sigma_2$-invariant, it is enough to show 
 the statement for a specific $k\in\Z$. So  let us consider  
\begin{align*}
O_2=
c_{3}^{\, \ast}
\left(1-\frac{1}{g}\left(c_{4}^{\, \ast}c_{5}+
c_{1}c_{2}^{\, \ast}\right)
+\frac{2}{g^2}c_{1}c_{2}^{\, \ast}
c_{4}^{\, \ast}c_{5}\right)\in \Al([1,2,3,4,5,6])_{-}.
\end{align*}
Then for  the identity \eqref{eq:delg-local} to be valid it is enough to take 
$M_0=0-2=-2$ and $N_0=6+2=8$. 
 We compute 
\begin{align*}
\delta_g(O_2)&=\left[\widetilde{Q}(g)_{[-1,8]},\; O_2\right]_{\gamma}\\
&=\left[
\sum_{k=0}^{4} \left(gc_{2k-1}+ c_{2k-1} c^{\ast}_{2k} c_{2k+1}\right),
\; O_2\right]_{\gamma}\ \  (9=-1)\\
&=\left[g \left(c_{-1}+c_{1}+c_{3}+c_{5}+c_{7}\right)+
\left( c_{-1} c^{\ast}_{0} c_{1}
+c_{1} c^{\ast}_{2} c_{3}
+c_{3} c^{\ast}_{4} c_{5}+
c_{5} c^{\ast}_{6} c_{7}+c_{7} c^{\ast}_{8} c_{-1}\right),
\; O_2 \right]_{\gamma}\\
&=\left[ g\left(c_{1}+c_{3}+c_{5}\right)+
\left( c_{-1} c^{\ast}_{0} c_{1}
+c_{1} c^{\ast}_{2} c_{3}
+c_{3} c^{\ast}_{4} c_{5}+
c_{5} c^{\ast}_{6} c_{7}\right),
\; O_2 \right]_{\gamma},
\end{align*}
where the identification $9=-1$ is made and the 
 graded-locality \eqref{eq:gloc} is noted.
Similarly, we can verify that 
\begin{align*}
\delta_g(O_2)&=\left[\widehat{Q}(g)_{[-1,7]},\; O_2\right]_{\gamma}\\
&=g\left[ \left(c_{-1}+c_{1}+c_{3}+c_{5}+c_{7}\right)+
\left( c_{-1} c^{\ast}_{0} c_{1}
+c_{1} c^{\ast}_{2} c_{3}
+c_{3} c^{\ast}_{4} c_{5}+
c_{5} c^{\ast}_{6} c_{7}\right),
\; O_2 \right]_{\gamma}\\
&=g\left[ \left(c_{1}+c_{3}+c_{5}\right)+
\left( c_{-1} c^{\ast}_{0} c_{1}
+c_{1} c^{\ast}_{2} c_{3}
+c_{3} c^{\ast}_{4} c_{5}+
c_{5} c^{\ast}_{6} c_{7}\right),
\; O_2 \right]_{\gamma}.
\end{align*}
By direct computation using the
  canonical anticommutation relations 
 \eqref{eq:CAR}
 we have 
\begin{align*}
&\delta_g(O_2)\\
&=\left[ g\left(c_{1}+c_{3}+c_{5}\right)+
\left( c_{-1} c^{\ast}_{0} c_{1}
+c_{1} c^{\ast}_{2} c_{3}
+c_{3} c^{\ast}_{4} c_{5}+
c_{5} c^{\ast}_{6} c_{7}\right),
\; c_{3}^{\, \ast}
\left(1-\frac{1}{g}\left(c_{4}^{\, \ast}c_{5}+
c_{1}c_{2}^{\, \ast}\right)
+\frac{2}{g^2}c_{1}c_{2}^{\, \ast}
c_{4}^{\, \ast}c_{5}\right)
 \right]_{\gamma}\\
&=\{gc_{3}+c_{1} c^{\ast}_{2} c_{3}
+c_{3} c^{\ast}_{4} c_{5}, \; c_{3}^{\, \ast}\}
\left(1-\frac{1}{g}\left(c_{4}^{\, \ast}c_{5}+
c_{1}c_{2}^{\, \ast}\right)
+\frac{2}{g^2}c_{1}c_{2}^{\, \ast}
c_{4}^{\, \ast}c_{5}\right) \\
&\ \ -c_{3}^{\, \ast}\delta_g
\left(1-\frac{1}{g}\left(c_{4}^{\, \ast}c_{5}+
c_{1}c_{2}^{\, \ast}\right)
+\frac{2}{g^2}c_{1}c_{2}^{\, \ast}
c_{4}^{\, \ast}c_{5}\right)
\\
&=
(g+c_{1} c^{\ast}_{2} + c^{\ast}_{4} c_{5})
\left(1-\frac{1}{g}\left(c_{4}^{\, \ast}c_{5}+
c_{1}c_{2}^{\, \ast}\right)
+\frac{2}{g^2}c_{1}c_{2}^{\, \ast}
c_{4}^{\, \ast}c_{5}\right)  -c_{3}^{\, \ast}\cdot 0
\\
&=g-\left(c_{4}^{\, \ast}c_{5}+c_{1}c_{2}^{\, \ast}\right)
+\frac{2}{g} c_{1}c_{2}^{\, \ast}c_{4}^{\, \ast}c_{5}
+(c_{1} c^{\ast}_{2} + c^{\ast}_{4} c_{5})
-2\times \frac{1}{g}c_{1}c_{2}^{\, \ast}c_{4}^{\, \ast}c_{5}  \\
&=g,
\end{align*}
where we have noted  
\begin{align*}
&\delta_g\left(1-\frac{1}{g}\left(c_{4}^{\, \ast}c_{5}+
c_{1}c_{2}^{\, \ast}\right)
+\frac{2}{g^2}c_{1}c_{2}^{\, \ast}
c_{4}^{\, \ast}c_{5}\right)\\
&=\left[ g\left(c_{1}+c_{3}+c_{5}\right)+
\left( c_{-1} c^{\ast}_{0} c_{1}
+c_{1} c^{\ast}_{2} c_{3}
+c_{3} c^{\ast}_{4} c_{5}+
c_{5} c^{\ast}_{6} c_{7}\right),
\; 
\left(1-\frac{1}{g}\left(c_{4}^{\, \ast}c_{5}+
c_{1}c_{2}^{\, \ast}\right)
+\frac{2}{g^2}c_{1}c_{2}^{\, \ast}
c_{4}^{\, \ast}c_{5}\right)
 \right]_{\gamma}\\
&=0.
\end{align*}
Analogously we obtain  \eqref{eq:Okg} for any $k\in\Z$.

Now take  any (not necessarily homogeneous) state $\omega$ of $\Al$.
Then for any $k\in\Z$
\begin{align}
\label{eq:omeOkg}
\omega\left(\delta_g(O_k)\right)=\omega(g1)=g.
\end{align}
Thus if $g\ne0$, then $\omega$ is not  invariant under $\delta_g$.
As $\omega$ is arbitrary, there exists  no 
invariant state under $\delta_g$ and hence 
SUSY is spontaneously broken for any $g\ne0$.
\end{proof}

\subsection{Positivity of energy density due to supersymmetry breakdown}
\label{subsec:}
  We shall discuss  the energy density for homogeneous  ground 
 states.
Let us fix relevant notation. 
A state $\omega$ on $\Al$ is called translation invariant 
if $\omega(A)=\omega\left(\sigma_{1}(A)\right)$ for all $A\in\Al$.
A state $\omega$ on $\Al$ is called homogeneous  (with periodicity $2$) 
if $\omega(A)=\omega\left(\sigma_{2}(A)\right)$ for all $A\in\Al$.

For any homogeneous state $\omega$ we can define the energy density 
 for the extended Nicolai model 
 by   the expectation value of the local Hamiltonians per site in the 
infinite-volume limit. As we can  choose any  boundary condition upon local Hamiltonians  as noted in  \cite{SIMON}, we have 
\begin{align}
\label{eq:eome}
e(g)(\omega):=\lim_{N\to \infty}
\frac{1}{2N}\omega\left(\widetilde{H}(g)_{[-N+1,N]}\right)
=\lim_{N\to \infty }
\frac{1}{2N}\omega \left(\widehat{H}(g)_{[-N+1,N+1]}\right).
\end{align}
 Since all $\widetilde{H}(g)_{[-N+1,N]}$ (as well as $\widehat{H}(g)_{[-N+1,N]}$) are positive operators 
 by definition, we have 
\begin{align}
\label{eq:eome-pos}
e(g)(\omega)\ge 0.
\end{align}

A  general  definition of ground states for $\cstar$-systems  is  
 given in terms of the infinitesimal time evolution, see \cite{BR}.  
 Now it is  $d_g$  given in \eqref{eq:del-SUSY-relations}.  
For  homogeneous states, 
 this general  characterization of ground states  is known to be  
 equivalent to  the minimum energy-density condition  \cite{GROUND}.
The extended Nicolai model on $\Z$ is a homogeneous  model of $2$-periodicity.
It has been known that for any  translation invariant model 
 there is at least one translation invariant (pure or non-pure) ground  state.
 Hence  there exists at least one 
 homogeneous ground state $\varphi$ of the periodicity $2$ 
 for the extended Nicolai model. 
The second theorem is as follows.
\begin{thm}
\label{thm:g-positive}
Let $\varphi$ be any homogeneous ground state for the extended Nicolai model. 
The energy density $e(g)(\varphi)$
 is strictly positive if  the parameter  $g$  of the model is not $0$.
\end{thm}

\begin{proof}
 Some idea of the  proof which we will present  below is  owing  to  Buchholz \cite{BUCHLNP} 
and Buchholz-Ojima \cite{BUOJ}.
We will  use the standard formulation of GNS representations, see  \cite{BR}. 
By  $\bigl(\Hilvp,\; \pivp,\; \Omevp  \bigr)$
 we denote the  GNS representation  associated to the state  $\varphi$ of $\Al$. Precisely,   
$\pivp$ is  a  homomorphism from $\Al$
 into $\Bl(\Hilvp)$ (the set of all bounded linear 
 operators on the Hilbert space $\Hilvp$), 
 and  $\Omevp\in\Hilvp$ is a cyclic vector such that 
$\varphi(A)=\left<\Omevp,\; \pivp(A)\Omevp\right>$  for all $A\in \Al$.

We  consider finite  averages of local operators $\{O_k \}$ 
 defined in \eqref{eq:Ok} under shift-translations:
For $n\in\NN$ let   
\begin{align}
\label{eq:defon}
o(n):=\frac{1}{n}
\sum_{k=1}^{n}O_k \in \Al([-1,2n+1])_{-}.
\end{align}
We shall study the asymptotic behavior of $\varphi\left(\delta_g(o(n))\right)$
 as $n\to \infty$.
By \eqref{eq:Okg} we have  
\begin{align}
\label{eq:delg-on=n}
\varphi\left(\delta_g(o(n))\right)=
\frac{1}{n}
\sum_{k=1}^{n}
\varphi\left(\delta_g(O_k)\right)=\frac{1}{n}
\sum_{k=1}^{n}g=g.
\end{align}

We can rewrite 
$\delta_g(o(n))\in\core$ in terms of  finite supercharges
 which are located in a slightly larger 
 region  including   the support region  of $o(n)$.
As in \eqref{eq:delg-local} by using local supercharges \eqref{eq:Qgper} 
 under  periodic boundary conditions we have   
\begin{equation}
\label{eq:delg-on-per}
\delta_g(o(n))=\left[\widetilde{Q}(g)_{[-3,2(n+2)]},\; o(n)\right]_{\gamma}.
\end{equation}
Similarly, we may  use   free-boundary supercharges  \eqref{eq:Qgfree} 
 as  nothing will change  due to  the choice of boundary conditions.
 By using  the GNS representation  $\bigl(\Hilvp,\; \pivp,\; \Omevp  \bigr)$
 we have 
\begin{align}
\label{eq:GNSestimate}
&\varphi\left(\delta_g(o(n))\right)\\
&=
\left<\Omevp, \pivp\left(\left[\widetilde{Q}(g)_{[-3,2(n+2)]},\; o(n)\right]_{\gamma}\right)\Omevp\right>\nonumber \\
&=
\left<\Omevp, \left( \pivp\left(\widetilde{Q}(g)_{[-3,2(n+2)]}\right)
\pivp(o(n))+\pivp(o(n)) \pivp\left(\widetilde{Q}(g)_{[-3,2(n+2)]}\right)
\right) \Omevp\right>\nonumber \\
&=
\left<\pivp\left(\widetilde{Q}(g)_{[-3,2(n+2)]}\right)^{\ast}\Omevp, 
\pivp(o(n))\Omevp\right>
+
\left<\pivp(o(n))^{\ast}\Omevp, 
 \pivp\left(\widetilde{Q}(g)_{[-3,2(n+2)]}\right)
\Omevp\right>.
\end{align}
As $\Omevp$ is a normalized vector, by using the triangle inequality and 
 Cauchy-Schwarz inequality 
this yields the following estimate 
\begin{align}
\label{eq:cru}
&\left|\varphi\left(\delta_g(o(n))\right)\right| \nonumber\\
&\le 
\left\Vert \pivp\left(\widetilde{Q}(g)_{[-3,2(n+2)]}\right)^{\ast}\Omevp \right\Vert 
\cdot \Vert
\pivp(o(n)) \Omevp\Vert +
 \Vert \pivp(o(n))^{\ast}\Omevp \Vert \cdot 
 \left \Vert \pivp\left(\widetilde{Q}(g)_{[-3,2(n+2)]}\right) \Omevp\right \Vert \nonumber
\\
&\le 
\left(
\left\Vert \pivp\left(\widetilde{Q}(g)_{[-3,2(n+2)]}\right)^{\ast}
\Omevp \right\Vert+
\left\Vert \pivp\left( \widetilde{Q}(g)_{[-3,2(n+2)]} \right) \Omevp\right\Vert 
\right)\cdot \Vert o(n)\Vert.
\end{align}
By applying  Lemma \ref{lem:onestimate}
and Lemma \ref{lem:QnOmevp} which will be shown later  into the above estimate \eqref{eq:cru} we obtain
\begin{align}
\label{eq:asymzero}
\lim_{n\to\infty} \left|\varphi\left(\delta_g(o(n))\right)\right| =0.
\end{align}
This contradicts with \eqref{eq:delg-on=n} when $g\ne 0$.
Thus when $g\ne 0$  the assumption of Lemma \ref{lem:QnOmevp}
 does not hold, and accordingly  $e(g)(\varphi)$ should be non-zero.
\end{proof}

Now we will show  two lemmas  used in  the above theorem.
We  recall the  Landau notation: 
$\BO$ is so called ``big-O'', 
 and $\lo$ is so called ``little-o''. 
 The following 
  mathematical  statement  
  claims   non-existence of  averaged   fermion operators (fermion observables at infinity)  
  in the infinite-volume limit. 
As the  estimate is  also  essential,
 we shall recapture    its derivation from  the original work \cite{LAN-ROB}.
\begin{lem}
\label{lem:onestimate}
\begin{align}
\label{eq:on-asymp}
\Vert o(n)\Vert \sim \mathrm{O}
 \left(\frac{1}{\sqrt{n}}\right)\  \text{as}\ 
 n\to\infty.
\end{align}
In particular, $\lim_{n\to\infty}o(n)=0$ in norm.
\end{lem}

\begin{proof}
For any $F\in \Al$ 
the  inequality $\Vert  F \Vert^2= \Vert F^{\ast} F \Vert\le  \Vert F^{\ast} F+F F^{\ast} \Vert $ holds.
By using this obvious inequality we obtain 
\begin{align}
\label{eq:on-estimateCAR}
\Vert o(n)\Vert^{2}\le 
\frac{1}{n^2}
\sum_{k=1}^{n} \sum_{k^{\prime}=1}^{n}
\Vert \{O_{k}^{\ast},\; O_{k^{\prime}}  \}\Vert.
\end{align}
Each term is estimated from the above by some constant: 
\begin{align}
\label{eq:each-estimate}
\Vert \{O_{k}^{\ast},\; O_{k^{\prime}}  \} \Vert
\le \Vert O_{k}^{\ast}  O_{k^{\prime}}   \Vert+ 
 \Vert   O_{k^{\prime}} O_{k}^{\ast}  \Vert
\le  2  \Vert O_{k}^{\ast} \Vert \cdot \Vert O_{k^{\prime}}   \Vert
=  2   \Vert O_{1}   \Vert^2\equiv C^2/5
\end{align}
By the 
 graded-locality \eqref{eq:gloc} 
 and the definition of  $O_k$  given in \eqref{eq:Ok} we have 
\begin{align}
\label{eq:vanish}
 \{O_{k}^{\ast},\; O_{k^{\prime}}  \}=0 \quad \text{if}\ |k-k^{\prime}|>2.
\end{align}
Thus for each fixed $k\in\{1,2,\cdots,n\}$ 
 there are at most five  $k^{\prime}\in\{1,2,\cdots,n\}$
such that $\{O_{k}^{\ast},\; O_{k^{\prime}}  \}$ does not vanish.
By applying \eqref{eq:each-estimate} \eqref{eq:vanish}
to \eqref{eq:on-estimateCAR} we  obtain 
\begin{align}
\label{eq:on-est2}
\Vert o(n)\Vert^{2}\le 
\frac{1}{n^2}
\sum_{k=1}^{n} 5 \times C^2/5=\frac{C^2}{n}.
\end{align}
it is equivalent  to 
$\Vert o(n)\Vert \le \frac{C}{\sqrt{n}}$ giving  \eqref{eq:on-asymp}.
\end{proof}

\begin{lem}
\label{lem:QnOmevp}
If the energy density $e(g)(\varphi)$ 
 is equal to $0$, then 
\begin{align}
\label{eq:QnOmevp}
\left \Vert \pivp\left(\widetilde{Q}(g)_{[-3,2(n+2)]}\right)^{\ast}\Omevp \right\Vert 
\sim \mathrm{o} \left({\sqrt{n}}\right),\ 
\left \Vert \pivp\left( \widetilde{Q}(g)_{[-3,2(n+2)]} \right) \Omevp\right \Vert 
\sim \mathrm{o} \left({\sqrt{n}}\right)\  \text{as}\ 
 n\to\infty.
\end{align}
\end{lem}
\begin{proof}
By the assumption $e(g)(\varphi)=0$  we have 
\begin{align}
\label{eq:density-zero}
0=e(g)(\varphi)=\lim_{n\to\infty }
\frac{1}{2n+8}\varphi\left(\widetilde{H}(g)_{[-3,2(n+2)]}\right)
=\lim_{n\to\infty }
\frac{1}{2n}\varphi\left(\widetilde{H}(g)_{[-3,2(n+2)]}\right).
\end{align}
By \eqref{eq:Hgper}
\begin{align}
\label{eq:}
&\varphi\left(\widetilde{H}(g)_{[-3,2(n+2)]}\right)\nonumber \\
&=\varphi\left( \Bigl\{\widetilde{Q}(g)_{[-3,2(n+2)]},\; 
{\widetilde{Q}(g)}_{[-3,2(n+2)]}^{\ast} \Bigr\}  \right)\nonumber \\
&=\left<\pivp\left({\widetilde{Q}(g)}_{[-3,2(n+2)]}^{\ast}\right)\Omevp, 
\pivp\left({\widetilde{Q}(g)}_{[-3,2(n+2)]}^{\ast}\right)\Omevp\right>
\nonumber \\
&\quad +\left<\pivp\left(\widetilde{Q}(g)_{[-3,2(n+2)]}\right)\Omevp, 
\pivp\left(\widetilde{Q}(g)_{[-3,2(n+2)]}\right)\Omevp\right> \nonumber \\
&=\left\Vert \pivp\left({\widetilde{Q}(g)}_{[-3,2(n+2)]}^{\ast}\right)\Omevp 
\right\Vert^2
+\left \Vert \pivp\left(\widetilde{Q}(g)_{[-3,2(n+2)]}\right)\Omevp 
 \right\Vert^2.
\end{align}
By this  together with
 \eqref{eq:density-zero} 
we have 
\begin{align}
\label{eq:}
\left \Vert \pivp\left({\widetilde{Q}(g)}_{[-3,2(n+2)]}^{\ast}\right)\Omevp 
 \right\Vert^2
\sim \mathrm{o} \left(2n\right),\quad 
\left\Vert \pivp\left(\widetilde{Q}(g)_{[-3,2(n+2)]} \right)\Omevp \right\Vert^2
\sim \mathrm{o} \left(2n\right).
\end{align}
From these  we obtain  \eqref{eq:QnOmevp}. 
\end{proof}

\begin{acknowledgments}
I  thank Hosho Katsura,  Yu Nakayama and Noriaki Sannomiya for 
 correspondence.
I thank  members  of  Kanazawa University for discussion.
\end{acknowledgments}
\bibliography{moriya-susy}
\end{document}